\documentclass[reprint,pra,twocolumn,showpacs]{revtex4-1}

\usepackage{amsmath}
\usepackage{amsthm}
\usepackage{latexsym}
\usepackage{amsfonts}
\usepackage{amssymb}
\usepackage{color}
\usepackage{bbm,dsfont}
\usepackage{graphicx}
\usepackage{hyperref}

\newtheorem{proposition}{Proposition}
\newtheorem{proposition?}{Proposition?}

\theoremstyle{definition}

\newtheorem{definition}{Definition}

\newcommand{\real}{\mathbb R} 
\newcommand{\complex}{\mathbb C} 

\newcommand{\hi}{\mathcal{H}} 
\newcommand{\tr}[1]{\textrm{tr}\left[#1\right]} 
\newcommand{\id}{\mathbbm{1}} 

\newcommand{\A}{\mathsf{A}}
\newcommand{\B}{\mathsf{B}}
\newcommand{\E}{\mathsf{E}}
\newcommand{\G}{\mathsf{G}}
\newcommand{\Q}{\mathsf{Q}}

\newcommand{\h}{\mathcal}

\begin{document}

\title[Breaking Gaussian incompatibility]{Breaking Gaussian incompatibility on continuous variable quantum systems}

\author{Teiko Heinosaari}
\email{teiko.heinosaari@utu.fi}
\affiliation{Turku Centre for Quantum Physics, Department of Physics and Astronomy, University of Turku, FI-20014 Turku, Finland}

\author{Jukka Kiukas}
\email{jukka.kiukas@nottingham.ac.uk}
\affiliation{School of Mathematical Sciences, University of Nottingham, University Park,
Nottingham, NG7 2RD, UK}

\author{Jussi Schultz}
\email{jussi.schultz@gmail.com}
\affiliation{Dipartimento di Matematica, Politecnico di Milano, Piazza Leonardo da Vinci 32, I-20133 Milano, Italy}
\affiliation{Turku Centre for Quantum Physics, Department of Physics and Astronomy, University of Turku, FI-20014 Turku, Finland}

\begin{abstract} We characterise Gaussian quantum channels that are {\it Gaussian incompatibility breaking}, that is, transform every set of Gaussian measurements into a set obtainable from a joint Gaussian observable via Gaussian postprocessing. Such channels represent local noise which renders measurements useless for Gaussian EPR-steering, providing the appropriate generalisation of entanglement breaking channels for this scenario. Understanding the structure of Gaussian incompatibility breaking channels contributes to the resource theory of noisy continuous variable quantum information protocols.
\end{abstract}

\maketitle

\section{Introduction}

Continuous variable systems constitute the basic framework for quantum optical applications \cite{BrLo05,WeEtal12}. Consequently, quantum information ideas are typically also discussed in that setting, despite the fact that the Hilbert space $\mathcal H$ is infinite-dimensional. 
The technical problems arising from the infinite setting become largely irrelevant when one restricts to the \emph{Gaussian} scenario, where the appropriate quantum objects (states, measurements, and channels) are simply described by finite matrices. In particular, EPR-steering for Gaussian states was characterised in the seminal article \cite{WiJoDo07}, subsequently giving rise to e.g. the recent studies \cite{KoLeRaAd15, JiKiNh15}.

In the finite-dimensional setting, it has been observed \cite{UoMoGu14,QuVeBr14} that the essential quantum measurement resource for steering consists of \emph{incompatible} observables; in this way, steering motivates the study of incompatibility in general. The purpose of the present paper is to consider \emph{incompatibility of Gaussian measurements} in noisy scenarios, where the noise is accordingly described by a Gaussian quantum channel. 
More specifically, we consider Gaussian channels which transform every set of Gaussian measurements into a set having a joint Gaussian observable.
Such channels will be called \emph{Gaussian incompatibility breaking}, in analogy to the general idea of incompatibility breaking channels introduced recently in \cite{HeKiReSc15} as an appropriate generalisation of entanglement breaking channels for e.g. steering scenarios.

From the technical point of view, the formulation of joint measurability in the Gaussian setting requires some care. On the one hand, this is due to Gaussian measurements having continuous outcome sets necessitating the description of these observables literally as operator valued measures, or the corresponding operator-valued characteristic functions. On the other hand, we need to consider compatibility of an infinite number of observables, which has to be defined in terms of appropriate postprocessing functions.

The main result of the paper is a characterisation of Gaussian incompatibility breaking channels; this takes the form of a matrix inequality involving the parameters of a Gaussian channel.

The structure of the paper is as follows: in section \ref{sec:gm} we first review the structure of Gaussian states, channels, and observables using Weyl operators and characteristic functions. Then we formulate compatibility of Gaussian measurements in terms of Gaussian joint observables and Gaussian postprocessings. Section \ref{sec:gibc} is devoted to the main new notion of the paper, Gaussian incompatibility breaking channels. In section \ref{sec:con} we investigate their relationship with entanglement breaking Gaussian channels, as well as the connection to the Gaussian steerability condition given in \cite{WiJoDo07}.

\section{Gaussian measurements and their Gaussian compatibility}\label{sec:gm}

\subsection{Continuous variable systems}

The description of a continuous variable quantum system with $2N$ canonical degrees of freedom is based on the $N$-fold tensor product Hilbert space $\hi^{\otimes N}  = \bigotimes_{j=1}^N L^2(\real) \simeq L^2(\real^N)$. We denote by $\h L(\hi^{\otimes N})$ and $\h T(\hi^{\otimes N})$ the Banach spaces of bounded and trace class operators on $\hi^{\otimes N}$, respectively. Furthermore, we denote by $Q_j$ and $P_j$ the canonical position and momentum operators acting on the $j$th component Hilbert space, and arrange them into a single vector 
$$
{\bf R} = (Q_1,P_1,\ldots, Q_N,P_N)^T \, .
$$
The phase space of the system is then $\real^{2N}$, and we denote the canonical coordinates by 
$$
{\bf x}=(q_1,p_1,\ldots, q_N,p_N)^T \, .
$$
The phase space translations are represented by the Weyl operators 
$$
W({\bf x}) = e^{-i {\bf x}^T {\bf \Omega R}} \quad\text{ where }\quad {\bf \Omega} = \bigoplus_{j=1}^N \left( \begin{array}{cc} 0 & 1 \\ -1 & 0\end{array}\right),
$$
and which satisfy the commutation relation 
\begin{equation}\label{eqn:commutation}
W({\bf x})W({\bf y}) = e^{-i{\bf x}^T {\bf \Omega y}} W({\bf y})W({\bf x}) \, .
\end{equation}
We refer to  \cite{BrLo05,WeEtal12,QSCI12} for more details on the description of continuous variable quantum systems.
 
\subsection{Gaussian states, channels, and observables}

The quantum states of a continuous variable system are in general represented by positive operators $\rho\in\h T(\hi^{\otimes N})$ with unit trace, and a state $\rho$ is called a \emph{Gaussian state} if its Weyl transform (or characteristic function) is a Gaussian function 
\begin{equation*}
\tr{\rho W({\bf x})} = e^{-\frac{1}{4}{\bf x}^T {\bf \Omega}^T{\bf V} {\bf \Omega}{\bf x} - i({\bf \Omega r})^T {\bf x}}  
\end{equation*}
where the \emph{covariance matrix} ${\bf V}$ is defined as 
\begin{equation*}
V_{ij} = \tr{\rho\{ R_i - r_i,R_j-r_j \}}
\end{equation*}
 with $r_j=\tr{\rho R_j}$ defining the displacement vector ${\bf r}$. 
The covariance matrix satisfies the uncertainty relation $
{\bf V} + i{\bf \Omega }\geq 0$, which is also a sufficient condition for a real symmetric matrix to be a valid covariance matrix \cite{SiMuDu94}.

A \emph{quantum channel} is defined in the Heisenberg picture as a completely positive normal linear map $\Lambda:\h L(\hi^{\otimes N'})\to\h L(\hi^{\otimes N})$ satisfying unitality $\Lambda(\id) =\id$. Normality, i.e. weak-* continuity, is required in the infinite-dimensional case to guarantee the existence of the Schr\"odinger picture version of the channel. This is given by the preadjoint map $\Lambda_* :\h T(\hi^{\otimes N})\to\h T(\hi^{\otimes N'})$ which is completely positive and trace preserving. The connection between the Heisenberg and Schr\"odinger pictures is given by the duality relation 
\begin{equation*}
\tr{\rho \Lambda(A)}=\tr{\Lambda_*(\rho)A} \, .
\end{equation*} 
Note that we allow the input and output spaces to be different, so that the channel may map a quantum system into another system having more or less degrees of freedom. A channel $\Lambda$ is called a \emph{Gaussian channel} if it maps Gaussian states into Gaussian states \cite{WeEtal12,GeCi02}. We describe them by their action on the Weyl operators as in \cite{DeVaVe77}, so that
\begin{equation}\label{eqn:gaussian_channel}
\Lambda(W({\bf x})) = W({\bf Ax})e^{-\frac{1}{4}{\bf x}^T {\bf Bx} - i{\bf c}^T{\bf x}}
\end{equation}
where ${\bf A}$ is a real $2N\times 2N'$ real matrix, and ${\bf B}$ is a real $2N'\times 2N'$-matrix which satisfy the complete positivity condition 
\begin{equation}\label{eqn:cp_channel}
{\bf B} + i {\bf \Omega} - i{\bf A}^T{\bf \Omega A} \geq {\bf 0} \, ,
\end{equation}
 whereas ${\bf c}\in\real^{2N'}$ can be arbitrary \cite{DeVaVe77, HoWe01, QSCI12}. Since a Gaussian channel is uniquely determined by the triple $({\bf A}, {\bf B}, {\bf c})$, we often denote $\Lambda= \Lambda_{({\bf A}, {\bf B}, {\bf c})}$. It should be noted that since we allow the input and output space of a channel to be different, the two ${\bf \Omega}$-matrices appearing in \eqref{eqn:cp_channel} typically operate on different phase spaces. However, since the phase space is always clear from the context, we refrain from indicating this explicitly.   

A \emph{quantum observable} is defined mathematically as a positive operator valued measure (POVM) $\E:\h B (\real^M)\to\h L(\hi^{\otimes N})$ (here $\h B(\real^M)$ denotes the Borel $\sigma$-algebra of $\real^M$), i.e., $\E(X)\geq 0$ for all $X$, $\E(\real^M) = \id$, and $\sigma$-additivity $\E(\cup_i X_i) = \sum_i \E(X_i)$ holds for all sequences of disjoint sets $X_i$.

An observable $\E$ is said to be {\em Gaussian} if a Gaussian state always yields a Gaussian measurement outcome distribution. This is the case whenever the  Fourier transform of $\E$ is given by
\begin{equation*}
\widehat{\E}({\bf p}) = \int e^{i{\bf p}^T{\bf x}}\, {\rm d}\E ({\bf x})= W({\bf Kp}) e^{-\frac{1}{4} {\bf p}^T{\bf Lp} - i{\bf m}^T{\bf p}},
\end{equation*}
where ${\bf K}$ is a $2N\times M$-matrix, and ${\bf L}$ is a $M\times M$-matrix satisfying ${\bf L} - i{\bf K}^T{\bf \Omega K}\geq {\bf 0}$ \cite{KiSc13}. Hence, an arbitrary Gaussian observable can be described by matrices as well.

An important example of a Gaussian observable is obtained by choosing $M=1$, ${\bf K}={\bf k}^T=(k_1,\ldots,k_{2N})^T$, ${\bf L}={\bf m}={\bf 0}$; then $\widehat{\E}(p)=e^{-i p {\bf k}^T{\bf \Omega R}}$, which is just the unitary family corresponding to the Hermitian operator 
$$
\sum_{j=1}^{N} (k_{2j-1} Q_j-k_{2j} P_j),
$$ 
which is a linear combination of the canonical quadrature operators. We call such an operator a \emph{generalised quadrature}. The canonical quadratures $Q_j$ taken together form the prototypical Gaussian observable, namely the \emph{canonical position observable} $\Q:\h B(\real^M)\to\h L(\hi^{\otimes M})$, acting as
\begin{equation*}
\left[ \Q(X)\psi\right]({\bf x}) = \chi_X({\bf x})\psi ({\bf x}) \, ,
\end{equation*}
where $\chi_X$ denotes the indicator function of the set $X$. Any Gaussian observable can always be obtained from the canonical position observable by applying a suitable Gaussian channel (see e.g. \cite{KiSc13}). Since dilations of Gaussian channels are well known \cite{CaEiGiHo08}, this immediately gives us a measurement dilation for a Gaussian observable $\E$: an auxiliary system is first prepared in a Gaussian state, after which the two systems are coupled by a unitary operator causing an affine symplectic transformation on the phase space of the composite system, and finally $\Q$ is measured on part of the transformed system (we refer to \cite{KiSc13} for more details). In quantum optical applications this has a clear physical meaning since the desired unitary coupling can be achieved by multi-port interferometry combined with squeezing, whereas the measurement of $\Q$ just corresponds to homodyne detection \cite{WeEtal12}.

A Gaussian channel transforms any Gaussian observable into another Gaussian observable. In fact, since
$$
\Lambda_{({\bf A}, {\bf B}, {\bf c})} (\widehat{\E}({\bf p})) = W({\bf AKp}) e^{-\frac{1}{4} {\bf p}^T({\bf L}+{\bf K}^T{\bf BK}){\bf p} - i({\bf m} +{\bf K}^T{\bf c})^T{\bf p}}
$$ 
we find that the action of the channel causes the transformation
\begin{equation}\label{eqn:gaussian_preprocessing}
({\bf K}, {\bf L},{\bf m})\mapsto ({\bf AK}, {\bf L} + {\bf K}^T  {\bf BK},{\bf m} + {\bf K}^T{\bf c})
\end{equation}
on the parameters of the observable. In particular, an initially projective (${\bf L}={\bf 0}$) Gaussian observable typically acquires additional noise in the form of ${\bf K}^T  {\bf BK}$, making it non-projective.

\subsection{Gaussian postprocessings and Gaussian compatibility}
A pair of (not necessarily Gaussian) observables $\E_1:\h B(\real^{M_1})\to\h L(\hi^{\otimes N})$ and $\E_2:\h B(\real^{M_2})\to\h L(\hi^{\otimes N})$ is  said to be {\em compatible}, or {\em jointly measurable}, if there exists a joint observable $\G:\h B(\real^{M_1+M_2})\to\h L(\hi^{\otimes N})$ having $\E_1$ and $\E_2$ as the (Cartesian) \emph{margins}:
\begin{equation*}
\E_1(X_1) = \G(X_1\times \real^{M_2}), \qquad \E_2(X_2) = \G(\real^{M_1}\times X_2). 
\end{equation*}
This definition can be naturally extended to any finite number of observables, but when dealing with the compatibility of an infinite set of observables one obviously runs into trouble. Therefore it is necessary to adjust this definition by allowing more general {\em postprocessing} of the joint observable, rather than just considering its margins. The appropriate mathematical generalization is given by Markov kernels: a Markov kernel is a map $f:\h B (\real^{M'})\times\real^M \to[0,1]$ such that $X\mapsto f(X,{\bf x})$ is a probability measure for all ${\bf x}\in\real^M$, and ${\bf x}\mapsto f(X,{\bf x})$ is measurable for all $X\in\h B(\real^{M'})$. An observable $\E$ on $\real^{M'}$ is then a postprocessing of another observable $\G$ on $\real^M$ if 
\begin{equation}\label{eqn:postprocessing}
\E(X) = \int f(X,{\bf x})\, {\rm d}\G({\bf x})
\end{equation}
for some postprocessing function, i.e., a Markov kernel $f$. A collection of observables $\h M$ is said to be \emph{compatible}, if there exists a joint observable $\G$ such that each $\E\in\h M$ is a postprocessing of $\G$ with some postprocessing function $f_\E$. It should be noted that for a finite set of observables, this definition of compatibility is equivalent to the more typical one involving only the Cartesian margins \cite{AlCaHeTo09}. In view of the hidden state models appearing in the steering context, the general definition via postprocessing functions \eqref{eqn:postprocessing} is more natural even in the finite setting; see \cite{HeKiReSc15} for discussion.

In this paper we are only interested in a fully Gaussian scenario; accordingly, we need to restrict the set of allowed postprocessing functions to those which transform Gaussian observables into other Gaussian ones. Since Gaussian observables are determined by the triples $({\bf K}, {\bf L}, {\bf m})$, a suitable postprocessing function should induce a transformation of these parameters.  If we calculate the Fourier transform of the observable $\E$ in Eq.~\eqref{eqn:postprocessing} we find that
$$
\widehat{\E}({\bf p})  = \int \left( \int e^{i{\bf p}^T{\bf y}} f({\rm d}{\bf y}, {\bf x}) \right) \, {\rm d}\G({\bf x}).
$$ 
In other words, by defining $w({\bf p})$ to be the character $w({\bf p})({\bf y}) = e^{i{\bf p}^T{\bf y}}$ we see that the postprocessing induces the transformation 
$$
w({\bf p}) \mapsto \lambda(w({\bf p})) = \int e^{i{\bf p}^T{\bf y}} f({\rm d}{\bf y}, \cdot). 
$$
In analogy with Eq.~\eqref{eqn:gaussian_channel}, we now define a {\em Gaussian postprocessing} to be one acting as a classical Gaussian channel:
\begin{equation*}
\lambda (w({\bf p})) = w({\bf Ap}) e^{-\frac{1}{4}{\bf p}^T {\bf Bp} - i{\bf c}^T{\bf p}}
\end{equation*}
for some $M'\times M$ matrix ${\bf A}$, a positive $M\times M$ matrix ${\bf B}$, and a vector ${\bf c}$. As with channels and observables, we use the triple notation $({\bf A}, {\bf B}, {\bf c})$ when referring to a Gaussian postprocessing.

If  $\G$ is a Gaussian observable determined by the triple $({\bf K}, {\bf L}, {\bf m})$, we have
\begin{eqnarray*}
 \widehat{\E}({\bf p})&=&\int \lambda(w({\bf p})) \, {\rm d}\G \\
 &=& W({\bf KAp})e^{-\frac{1}{4}{\bf p}^T({\bf B} + {\bf A}^T {\bf LA}) {\bf p} - i({\bf c + {\bf A}^T {\bf m})^T{\bf p}}} \, ,
\end{eqnarray*}
so that the Gaussian postprocessing with parameters $({\bf A}, {\bf B}, {\bf c})$ indeed induces the transformation
\begin{equation}\label{eqn:gaussian_postprocessing}
({\bf K}, {\bf L}, {\bf m})\mapsto ({\bf KA}, {\bf B} + {\bf A}^T{\bf L}{\bf A}, {\bf c} + {\bf A}^T{\bf m})
\end{equation}
on the parameters of the observable; this should be compared with Eq.~\eqref{eqn:gaussian_preprocessing}. We now see that this transformation can be obtained by combining two simple postprocessings of the form \eqref{eqn:postprocessing}: coordinate transformations and convolution type smearing. Firstly, if ${\bf A}$ is  a $M'\times M$ matrix, then we can define the Markov kernel $f_{\bf A} (X, {\bf x}) = \chi_X({\bf A}^T {\bf x})$  which, when applied to a Gaussian observable, induces the transformation
$$
({\bf K}, {\bf L}, {\bf m})\mapsto ({\bf KA}, {\bf A}^T{\bf LA}, {\bf A}^T{\bf m})
$$  
Secondly, if $\mu$ is a Gaussian probability measure on $\real^M$ with $\widehat{\mu}({\bf p}) = e^{-\frac{1}{4}{\bf p}^T {\bf Bp} - i {\bf c}^T{\bf p}}$ for a positive matrix ${\bf B}$, then the Markov kernel $f_\mu(X,{\bf x}) = \mu (X-{\bf x})$ defines a smearing of $\E$ which we denote by $\mu*\E$. Since the Fourier transform maps convolutions into products, we notice that the induced transformation on the parameters is 
$$
({\bf K}, {\bf L}, {\bf m})\mapsto ({\bf K}, {\bf B} + {\bf L}, {\bf c} + {\bf m}).
$$
Hence, by applying these two natural postprocessings in sequel, we obtain the desired general postprocessing transformation \eqref{eqn:gaussian_postprocessing}. We are now ready to define compatibility in the Gaussian setting:
\begin{definition}\label{def:gcomp}
A collection $\h M$ of Gaussian observables is {\em Gaussian compatible} if there exists a Gaussian observable $\G$ such that any $\E\in\h M$ can be obtained from $\G$ via Gaussian postprocessing.
\end{definition}

The next result shows that for a finite collection of Gaussian observables, Gaussian compatibility is also equivalent to joint measurability in the usual sense, i.e. the existence of a joint Gaussian observable which gives the original observables as its margins.

\begin{proposition}\label{prop:finite}
Let $\E_j:\h B(\real^{M_j})\to\h L(\hi^{\otimes N})$ be a Gaussian observable for each $j=1,\ldots, n$. Then the collection $\{\E_1,\ldots, \E_n\}$ is Gaussian compatible if and only if there exists a Gaussian observable $\G:\h B(\real^M)\to\h L(\hi^{\otimes N})$, where $M= \sum_{j=1}^n M_j$, which gives the $\E_j$:s as its margins.
\end{proposition}

\begin{proof}
Suppose first that the $\E_j$:s have a Gaussian joint observable $\G$ on $\real^M$, and let $\G$ be determined by  the triple $({\bf K}, {\bf L}, {\bf m})$. We can write these in block form as
$$
{\bf K} = \left( 
{\bf K}_1, \ldots , {\bf K}_n
\right), \quad
{\bf L} = \left( \begin{array}{ccc}   
{\bf L}_{11} & \cdots & {\bf L}_{1n}\\
\vdots & \ddots & \vdots\\
{\bf L}_{n1} & \cdots & {\bf L}_{nn}
\end{array}\right), 
$$
and 
$$
{\bf m} = \left( \begin{array}{c}   
{\bf m}_1 \\
 \vdots \\
  {\bf m}_n
\end{array}\right)
$$
where now ${\bf K}_j$, ${\bf L}_{jj}$, and ${\bf m}_j$ are the parameters of the $j$th margin. In other words, taking the $j$th margin corresponds to a Gaussian postprocessing with $({\bf A}_j, {\bf 0}_j, {\bf 0}_j)$ where ${\bf A}_j$ is the $M\times M_j$ matrix 
$$
{\bf A}_j = \left( \begin{array}{c}   
{\bf 0}_{M_1\times M_j} \\
 \vdots \\
 {\bf 0}_{M_{j-1}\times M_j}\\
  {\bf I}_{M_j\times M_j}\\
  {\bf 0}_{M_{j+1}\times M_j}\\
  \vdots\\
  {\bf 0}_{M_{n}\times M_j}
\end{array}\right)
$$
Hence, the observables are Gaussian compatible in the sense of Def. \ref{def:gcomp}. 

Suppose then  that the $\E_j$ are Gaussian compatible, and obtained from a Gaussian observable with parameters $({\bf K}, {\bf L}, {\bf m})$ via Gaussian postprocessings given by $({\bf A}_j, {\bf B}_j, {\bf c}_j)$. This means that 
$$
({\bf K}_j, {\bf L}_j, {\bf m}_j) =({\bf K}{\bf A}_j, {\bf B}_j + {\bf A}_j^T{\bf L}{\bf A}_j, {\bf c}_j + {\bf A}_j^T{\bf m}) 
$$
for each $j$. We can now define the matrices ${\bf K}_0= ({\bf KA}_1,  \ldots,  {\bf KA}_n)$ and ${\bf L}_0 = {\rm diag} ({\bf B}_1 + {\bf A}_1^T {\bf LA}_1 , \ldots, {\bf B}_n + {\bf A}_n^T {\bf LA}_n )$, and the vector 
$$
{\bf m}_0  = \left( \begin{array}{c}   
{\bf c}_1 + {\bf A}_1^T{\bf m} \\
 \vdots \\
{\bf c}_n + {\bf A}_n^T{\bf m} 
\end{array}\right).
$$
What remains is to show that these parameters determine a valid Gaussian observable. 
By defining 
$$
{\bf A} = ({\bf A}_1, \ldots, {\bf A_n}) \, , \qquad {\bf B} = {\rm diag}({\bf B}_1, \ldots, {\bf B}_n) \, , 
$$
 we have that ${\bf K}_0 = {\bf KA}$ and ${\bf L}_0 = {\bf B} + {\bf A}^T {\bf LA}$. This implies directly that 
\begin{equation*}
 {\bf L}_0 + i{\bf K}_0^T {\bf \Omega}{\bf K}_0 = {\bf B} + {\bf A}^T({\bf L} + i{\bf K}^T {\bf \Omega}{\bf K}){\bf A}\geq 0 \, .
\end{equation*} 
In conclusion, the Gaussian observable determined by $({\bf K}_0, {\bf L}_0, {\bf m}_0)$ is the desired joint observable for the $\E_j$.
\end{proof}

\section{Gaussian incompatibility breaking channels}\label{sec:gibc}

We now proceed to define the central concept of \cite{HeKiReSc15} in the Gaussian setting:

\begin{definition}
We say that a Gaussian channel $\Lambda$ \emph{breaks the Gaussian incompatibility} of a subset $\mathcal M$ of Gaussian observables if the set $\Lambda(\mathcal M)$ is Gaussian compatible. We say that $\Lambda$ is {\em Gaussian incompatibility breaking} if it breaks the Gaussian incompatibility of the set of all Gaussian observables. 
\end{definition}

The following proposition is the main result of the paper:
\begin{proposition}\label{prop:gaussian_ibc}
A Gaussian channel $\Lambda_{({\bf A},{\bf B},{\bf c})}$ is Gaussian incompatibility breaking if and only if 
\begin{equation*}
{\bf B} - i{\bf A}^T{\bf \Omega} {\bf A}\geq {\bf 0}.
\end{equation*}
\end{proposition}
\begin{proof}
Suppose first that ${\bf B} - i{\bf A}^T{\bf \Omega} {\bf A}\geq {\bf 0}$. This implies that we can define a Gaussian observable $\G:\h B(\real^{2N'})\to\h L(\hi^{\otimes N})$ with the parameters $({\bf A}, {\bf B}, {\bf 0})$, that is, 
\begin{equation*}
\widehat{\G}({\bf p}) = W({\bf Ap}) e^{-\frac{1}{4}{\bf p}^T {\bf Bp}}.
\end{equation*}
Let $\E:\h B(\real^M)\to\h L(\hi^{\otimes N'})$ be a Gaussian observable with the parameters $({\bf K}, {\bf L},{\bf m})$. 
Since ${\bf L} - i{\bf K}^T{\bf \Omega K}$ is positive, it is in particular selfadjoint, so that 
\begin{equation*}
({\bf L} - i{\bf K}^T{\bf \Omega K})^* = {\bf L}^T - i{\bf K}^T{\bf \Omega K} = {\bf L} - i{\bf K}^T{\bf \Omega K}
\end{equation*}
(note that ${\bf \Omega}^* ={\bf \Omega}^T= -{\bf \Omega }$). This implies that ${\bf L}^T = {\bf L}$. Furthermore, since transposition is a positive map, we have that ${\bf L} + i{\bf K}^T{\bf \Omega K}\geq {\bf 0}$. These together imply that ${\bf L}\geq {\bf 0}$ and we can therefore define a Gaussian postprocessing with the parameters $({\bf K},{\bf L}, {\bf m} + {\bf K}^T{\bf c})$. By applying this postprocessing to the observable $\G$, we obtain the Gaussian observable with parameters 
$$
({\bf AK}, {\bf L} + {\bf K}^T{\bf B} {\bf K},  {\bf m} + {\bf K}^T{\bf c})
$$
which coincide with the parameters of $\Lambda_{({\bf A},{\bf B},{\bf c})} (\E)$ by Eq.~\eqref{eqn:gaussian_preprocessing}. Hence, $\Lambda_{({\bf A},{\bf B},{\bf v})}$ is Gaussian incompatibility breaking.

Suppose then that ${\bf B} - i{\bf A}^T{\bf\Omega} {\bf A}\geq {\bf 0}$ does not hold, i.e., there exists a ${\bf z}\in\complex^{2N'}$ such that 
\begin{equation*}
\overline{{\bf z}}^T  ({\bf B} - i{\bf A}^T{\bf\Omega} {\bf A}) {\bf z} <0 \, .
\end{equation*}
Now ${\bf A}$ and ${\bf B}$ still need to satisfy the condition ${\bf B} + i {\bf \Omega} - i{\bf A}^T{\bf \Omega A} \geq {\bf 0}$ which, again by taking the transpose, implies that ${\bf B}\geq {\bf 0}$. If we now write ${\bf z} ={\bf x}+ i {\bf y}$ with ${\bf x},{\bf y}\in\real^{2N'}$ and use the fact that 
\begin{equation*}
{\bf x}^T{\bf A}^T{\bf\Omega} {\bf A}{\bf x} = {\bf y}^T{\bf A}^T{\bf\Omega} {\bf A}{\bf y}=0 \, , 
\end{equation*}
 we find that both ${\bf x}$ and ${\bf y}$ must be nonzero, as otherwise we would have $\overline{{\bf z}}^T {\bf Bz} <0$. Furthermore, since ${\bf B}^T = {\bf B}$, we have that 
$$
 \overline{{\bf z}}^T  ({\bf B} - i{\bf A}^T{\bf\Omega} {\bf A}) {\bf z}  = {\bf x}^T{\bf Bx} + {\bf y}^T{\bf By} + 2{\bf x}^T{\bf A}^T{\bf \Omega A y}<0
$$
so that 
\begin{equation*}
-{\bf x}^T{\bf A}^T{\bf \Omega A y} > \frac{1}{2}({\bf x}^T{\bf Bx} + {\bf y}^T{\bf By}) \geq 0
\end{equation*}
by the positivity of ${\bf B}$. This gives us
\begin{equation*}
\left({\bf x}^T{\bf A}^T{\bf \Omega A y}\right)^2 >\frac{1}{4}\left({\bf x}^T{\bf Bx} + {\bf y}^T{\bf By}\right)^2 
\end{equation*}

We can now define the two Gaussian observables $\E_1,\E_2$ with outcome set $\real$, with the parameters $({\bf x}, {\bf 0}, {\bf 0})$ and $({\bf y}, {\bf 0}, {\bf 0})$, which are then transformed into $({\bf Ax}, {\bf x}^T{\bf Bx}, {\bf c})$ and $({\bf Ay}, {\bf y}^T{\bf By}, {\bf c})$, respectively, by  the channel $\Lambda_{({\bf A},{\bf B},{\bf c})}$. 
We claim that the corresponding observables $\Lambda_{({\bf A},{\bf B},{\bf c})}(\E_1)$ and $\Lambda_{({\bf A},{\bf B},{\bf c})}(\E_2)$ are Gaussian incompatible. 
Suppose, on the contrary, that they are Gaussian compatible. Then by Prop.\ref{prop:finite} they would also have  a Gaussian joint observable $\G$ with outcome set $\real^2$, from which $\Lambda_{({\bf A},{\bf B},{\bf c})}(\E_j)$ are obtained as margins. Let $\G$ be determined by the parameters $({\bf K}, {\bf L}, {\bf m})$. The marginal condition gives us 
$$
{\bf K} = \left(  {\bf Ax}, \, {\bf Ay}\right), \quad{\bf L} = \left( \begin{array}{cc} {\bf x}^T {\bf Bx} & l_{12} \\ l_{12} & {\bf y}^T{\bf By} \end{array}\right), 
$$
and
$$ 
{\bf m}=\left( \begin{array}{c}{\bf c} \\ {\bf c} \end{array}\right),
$$  
where we have used the fact that ${\bf L}^T = {\bf L}$.  We therefore have 
$$
{\bf L} - i{\bf K}^T{\bf \Omega K} = \left( \begin{array}{cc} {\bf x}^T {\bf Bx} & l_{12} - i{\bf x}^T {\bf A}^T{\bf \Omega A}{\bf y} \\ l_{12} + i{\bf x}^T {\bf A}^T{\bf \Omega A}{\bf y}& {\bf y}^T{\bf By} \end{array}\right)
$$
but this implies that 
\begin{align*}
&\det({\bf L} - i{\bf K}^T{\bf \Omega K})  \\
&\qquad= ({\bf x}^T {\bf Bx})( {\bf y}^T{\bf By}) - l_{12}^2 - ( {\bf x}^T {\bf A}^T{\bf \Omega A}{\bf y})^2\\
&\qquad < ({\bf x}^T {\bf Bx})( {\bf y}^T{\bf By}) - l_{12}^2 -\frac{1}{4}\left({\bf x}^T{\bf Bx} + {\bf y}^T{\bf By}\right)^2\\
&\qquad= -l_{12}^2 - \frac{1}{4} \left(  {\bf x}^T{\bf Bx} - {\bf y}^T{\bf By}\right)^2\\
&\qquad\leq 0.
\end{align*}
In other words, ${\bf L} - i{\bf K}^T{\bf \Omega K}$ is not positive which is a contradiction. Hence, $\Lambda_{({\bf A},{\bf B},{\bf c})}(\E_1)$ and $\Lambda_{({\bf A},{\bf B},{\bf c})}(\E_2)$  are Gaussian incompatible which proves that $\Lambda_{({\bf A},{\bf B},{\bf c})}$ is not Gaussian incompatibility breaking.
\end{proof}

The two observables $\E_1$ and $\E_2$ constructed in the proof of Prop.~\ref{prop:gaussian_ibc} are generalised quadratures.
As a side result we hence also obtain the following.  

\begin{proposition}\label{prop:equivalent}
For a Gaussian channel $\Lambda$, the following conditions are equivalent:
\begin{itemize}
\item[{\rm (i)}] $\Lambda$ is Gaussian incompatibility breaking.
\item[{\rm (ii)}] $\Lambda$ breaks the Gaussian incompatibility of each pair of Gaussian observables.
\item[{\rm (iii)}] $\Lambda$ breaks the Gaussian incompatibility of the set of all generalised quadratures.
\item[{\rm (iv)}] $\Lambda$ breaks the Gaussian incompatibility of each pair of generalised quadratures.
\end{itemize}
\end{proposition}

We recall that in the case of a finite dimensional Hilbert space, it is possible that a channel breaks the incompatibility of all pairs of observables but still does not break the incompatibility of a larger set of observables \cite{HeKiReSc15}. 
From this point of view Prop. \ref{prop:equivalent} reveals a qualitative difference between finite dimensional and Gaussian cases, reflecting the fact that Gaussian observables are very special.

\section{Connection to entanglement breaking channels and EPR-steering}\label{sec:con}
In this section we explicitly demonstrate that in the context of noisy Gaussian EPR-steering, the notion of Gaussian incompatibility breaking channels is exactly the appropriate generalisation of entanglement breaking channels.

\subsection{Entanglement breaking Gaussian channels}

Recall that a quantum channel $\Lambda$ is called entanglement breaking, if the bipartite state $(\Lambda_*\otimes {\rm Id} )(\rho)$ is separable for all $\rho$. 
It has been shown in \cite{Holevo08} that a Gaussian channel $\Lambda_{({\bf A},{\bf B},{\bf c})}$ is entanglement breaking if and only if  ${\bf B}$ can be decomposed as 
\begin{equation*}
{\bf B}={\bf B}_1 + {\bf B}_2  \quad \textrm{with} \quad  {\bf B}_1 + i{\bf \Omega}\geq {\bf 0} \, , \quad {\bf B}_2 - i {\bf A}^T{\bf \Omega A} \geq {\bf 0} \, .
\end{equation*}
It is now evident from this that such a channel is necessarily also Gaussian incompatibility breaking: the condition ${\bf B}_1 + i{\bf \Omega}\geq {\bf 0}$ implies that ${\bf B}_1\geq {\bf 0}$ which, combined with the second inequality, gives us 
\begin{equation*}
{\bf B} - i{\bf A}^T{\bf \Omega A}= {\bf B}_1 + {\bf B}_2 - i{\bf A}^T{\bf \Omega A}  \geq {\bf 0} \, .
\end{equation*}
This is exactly the condition stated in Prop. \ref{prop:gaussian_ibc}.

However, just as in the finite dimensional setting \cite{HeKiReSc15}, the converse implication does not hold in the Gaussian scenario. As a trivial example, consider the Gaussian channel $\Lambda_{({\bf I},{\bf I},{\bf 0})}$. This channel is Gaussian incompatibility breaking since ${\bf I} - i{\bf \Omega}\geq {\bf 0}$. If $\Lambda_{({\bf I},{\bf I},{\bf 0})}$ were entanglement breaking, then we would have a decomposition ${\bf I} = {\bf B}_1 + {\bf B}_2$ where ${\bf B}_1 + i{\bf \Omega}\geq {\bf 0}$ and ${\bf B}_2 - i{\bf \Omega}\geq {\bf 0}$. The latter inequality is equivalent to ${\bf B}_2 + i{\bf \Omega}\geq {\bf 0}$, hence we would obtain
\begin{equation*}
{\bf B}_1 + {\bf B}_2 + 2i{\bf \Omega} = {\bf I} + 2i{\bf \Omega}\geq {\bf 0} \,.
\end{equation*}

The above discussion is summarised in the following proposition.
\begin{proposition}
Every entanglement breaking Gaussian channel is Gaussian incompatibility breaking. The converse does not hold.
\end{proposition}

As a physically relevant example, let us now consider a Gaussian channel $\Lambda_{\bf B} = \Lambda_{({\bf I},{\bf B},{\bf 0})}$, in which case the only requirement for this to be a valid channel is the positivity of  ${\bf B}$. It follows that there exists a Gaussian probability measure $\mu$ such that $\widehat{\mu} ({\bf \Omega}^T{\bf x}) = e^{-\frac{1}{4}{\bf x}^T{\bf B x}}$. Using the commutation relation \eqref{eqn:commutation}, we see that 
\begin{align*}
\Lambda_{{\bf B}}(W({\bf x})) &= \widehat{\mu} ({\bf \Omega}^T{\bf x}) W({\bf x}) = \int e^{-i{\bf y}^T {\bf \Omega x}} \, d\mu({\bf y}) \, W({\bf x}) \\
&= \int W({\bf y})W({\bf x})W({\bf y})^*\, d\mu({\bf y})
\end{align*}
so that 
\begin{equation*}
\Lambda_{{\bf B}} (A) =  \int W({\bf y})AW({\bf y})^*\, d\mu({\bf y}) \, .
\end{equation*}
Such channels are sometimes called \emph{classical noise channels}.

Now the necessary and sufficient condition for $\Lambda_{{\bf B}}$ to be Gaussian incompatibility breaking reduces to ${\bf B} - i{\bf \Omega}\geq {\bf 0}$ which is equivalent to ${\bf B}$ being a valid covariance matrix of some Gaussian state $\rho$. In other words, $\tr{\rho W({\bf x})} = \widehat{\mu}({\bf \Omega}^T {\bf x})$ and since the Weyl transform is just the Fourier transform of the Wigner function, we find that for a Gaussian incompatibility breaking classical noise channel, the noise is always given by the (necessarily positive) Wigner function of a Gaussian state. 

In comparison, the channel $\Lambda_{\bf B} $ is entanglement breaking if and only if ${\bf B} = {\bf B}_1 + {\bf B}_2$ where ${\bf B}_1 +i{\bf \Omega}\geq {\bf 0}$ and ${\bf B}_2 - i{\bf \Omega}\geq {\bf 0}$. In other words, both ${\bf B}_1$ and ${\bf B}_2$ must be covariance matrices of some Gaussian states $\rho_1$ and $\rho_2$, respectively. But this means that $\tr{\rho W({\bf x})} = \tr{\rho_1W({\bf x})} \tr{\rho_2W({\bf x})}$, and since the Fourier transform maps products into convolutions, this implies that the Wigner function of $\rho$ can be written as the convolution of the Wigner functions of $\rho_1$ and $\rho_2$.

\subsection{Gaussian steering}
We now make a connection to known results on steerability of Gaussian states. In order to do that, we first need to formulate the EPR-steering scenario for observables having infinite number of outcomes.

The starting point is that of a correlation experiment where two parties, Alice and Bob, share a bipartite state $\rho$ and measure some observables $\A\in\h M_A$ and $\B\in\h M_B$ on their respective subsystems. The correlation table $\tr{\rho \A(X) \otimes \B(Y)}$ is said to  have a \emph{local classical model}, if there exist two families of Markov kernels $\{f_{\A}\}_{\A\in \mathcal M_A}$ and $\{g_{\B}\}_{\B\in \mathcal M_B}$ on a common probability space $(\Omega, \lambda)$, such that 
\begin{equation*}
{\rm tr}[\rho \A(X)\otimes \B(Y)]=\int_{\Omega} f_{\A}(X,\omega)g_{\B}(Y,\omega) {\rm d}\lambda(\omega),
\end{equation*}
for all $\A\in \mathcal M_A$ and $\B\in \mathcal M_B$. A {\em hidden state model} (on Bob's state space) is one for which $g_{\B}(Y,\omega)={\rm tr}[\rho(\omega) \B(Y)]$ for some (measurable) family of "hidden" states $\omega\mapsto \rho(\omega)$. If such a hidden state model does not exist, then it is said that Alice can \emph{steer} Bob's system, or that the state $\rho$ is {\em steerable} with Alice's measurements $\h M_A$.

For finite-dimensional systems, it has been shown \cite{UoMoGu14,QuVeBr14} that a pure state $\rho_0$ with full Schmidt rank is steerable with Alice's measurements $\h M_A$, if and only if $\h M_A$ is incompatible. Let then $\Lambda$ be a quantum channel. By the duality between states and measurements, the transformed state $(\Lambda_*\otimes {\rm Id})(\rho_0)$ is steerable with $\h M_A$ if and only if $\Lambda(\h M_A)$ is incompatible. From our point of view \cite{HeKiReSc15}, this means that the channel $\Lambda$ is incompatibility breaking, if and only if $(\Lambda_*\otimes {\rm Id})(\rho_0)$ is non-steerable by the total set of all measurements, and in that case, $(\Lambda_*\otimes {\rm Id})(\rho)$ is actually non-steerable for any state $\rho$. These general notions can also be extended to infinite-dimensional setting without much trouble; one merely needs to replace the maximally entangled state with the family $\rho_r$ of regularised EPR-states, i.e., Gaussian pure states with the covariance matrices 
$$
{\bf V}_0(r) = \left( \begin{array}{cc}
\cosh r {\bf I} & \sinh r{\bf Z} \\
\sinh r{\bf Z} & \cosh r {\bf I}
\end{array}\right) , \quad\text{ where }{\bf Z} = \bigoplus_{j=1}^N \sigma_z \, , 
$$
and do the corresponding regularisation for Alice's measurements. However, in this paper we do not need the general connection between incompatibility and steerability. Instead, we establish a fundamental relation between Gaussian incompatibility breaking quantum channels and steerability in the Gaussian setting. Concerning the latter, Wiseman et al.~have shown \cite{WiJoDo07} that a bipartite Gaussian state $ \rho$ with covariance matrix ${\bf V}$ is \emph{not} steerable with Alice's Gaussian measurements if and only if
$${\bf V} + i ({\bf 0}\oplus {\bf \Omega})\geq {\bf 0}.$$ 
In particular, the above EPR-states $\rho_r$ are all steerable; we again refer to \cite{WiJoDo07} for discussion on the original "EPR-paradox" in this context.

Suppose now that we begin with a bipartite Gaussian state $ \rho$ with covariance matrix ${\bf V}_0$, and subject it to a Gaussian channel $\Lambda_{({\bf A},{\bf B},{\bf c})}$ on Alice's side. 
The resulting state is then Gaussian with the covariance matrix
$$
{\bf V}={\bf \Omega} ({\bf A}\oplus {\bf I})^T {\bf \Omega}^T{\bf V}_0{\bf \Omega} ({\bf A}\oplus {\bf I} ){\bf \Omega}^T+{\bf \Omega}({\bf B}\oplus {\bf 0} ){\bf \Omega}^T.
$$
If it happens that ${\bf V} + i ({\bf 0}\oplus {\bf \Omega})\geq {\bf 0}$, then the final state is non-steerable and we may say that the channel $\Lambda_{({\bf A},{\bf B},{\bf c})}$ has \emph{broken the steerability} of $\rho$. If $\Lambda_{({\bf A},{\bf B},{\bf c})}$ breaks the steerability of any bipartite Gaussian state, then $\Lambda_{({\bf A},{\bf B},{\bf c})}$ may be called \emph{Gaussian steerability breaking}. With this terminology, we now have the following result.

\begin{proposition} Let $\Lambda$ be a Gaussian channel. The following conditions are equivalent:
\begin{itemize}
\item[(i)] $\Lambda$ is Gaussian incompatibility breaking;
\item[(ii)] $\Lambda$ is Gaussian steerability breaking;
\item[(iii)] $\Lambda$ breaks the steerability of all EPR-states $\rho_r$. 
\end{itemize}
\end{proposition}
\begin{proof}
We let $\Lambda=\Lambda_{({\bf A},{\bf B},{\bf c})}$ as above, and begin by looking at the matrix appearing in the non-steerability condition:
\begin{align*}
{\bf V}+i({\bf 0}\oplus {\bf \Omega}) &= {\bf \Omega} ({\bf A} \oplus{\bf I})^T{\bf\Omega}^T({\bf V}_0+i{\bf \Omega}){\bf \Omega}({\bf A}\oplus {\bf I}) {\bf \Omega}^T\\
&\quad+{\bf \Omega} (( {\bf B}-i{\bf A}^T{\bf \Omega} {\bf A})\oplus {\bf 0}){\bf \Omega}^T.
\end{align*}
The first term on the right-hand-side is positive since ${\bf V}_0+i{\bf \Omega}\geq {\bf 0}$ by the fact that ${\bf V}_0$ is a covariance matrix. If $\Lambda_{({\bf A},{\bf B},{\bf c})}$ is Gaussian incompatibility breaking, then by Prop. \ref{prop:gaussian_ibc} we have  ${\bf B}-i{\bf A}^T{\bf \Omega} {\bf A} \geq {\bf 0}$ and hence also the second term is positive. In other words, (i) implies (ii).

Trivially, (ii) implies (iii). In order to prove that (iii) implies (i), we consider the EPR-states $\rho_r$ defined above.
Using the fact that ${\bf \Omega}{\bf Z}{\bf \Omega}^T = - {\bf Z}$ we see that the covariance matrices of the transformed states are given by 
\begin{align*}
{\bf V}(r)=&\left(\begin{array}{cc} \cosh r\, {\bf \Omega}{\bf A}^T{\bf A\Omega}^T  &  -\sinh r \, {\bf \Omega }{\bf A}^T{\bf Z}{\bf \Omega}^T \\
-\sinh r \, {\bf \Omega ZA}{\bf \Omega}^T & \cosh r\, {\bf I}\end{array}\right)\\
& + \left(\begin{array}{cc}  {\bf \Omega B}{\bf \Omega}^T  & {\bf 0}\\ {\bf 0} & {\bf 0} \end{array}\right).
\end{align*} 
By denoting 
$$
{\bf C} = \left(\begin{array}{cc}  {\bf \Omega}{\bf A}^T{\bf A\Omega}^T &  - {\bf \Omega }{\bf A}^T{\bf Z}{\bf \Omega}^T \\
-{\bf \Omega ZA}{\bf \Omega}^T &  {\bf I}\end{array}\right) 
$$
and 
$$
{\bf C}' = \left(\begin{array}{cc}  {\bf \Omega}{\bf A}^T{\bf A\Omega}^T &   {\bf \Omega }{\bf A}^T{\bf Z}{\bf \Omega}^T \\{\bf \Omega ZA}{\bf \Omega}^T &  {\bf I}\end{array}\right) 
$$
we can write this as 
\begin{equation*}
{\bf V}(r) = \frac{1}{2}e^r {\bf C} + \frac{1}{2}e^{-r} {\bf C}' + {\bf \Omega B}{\bf \Omega}^T \oplus {\bf 0} \, .
\end{equation*}

Assuming that $\Lambda_{({\bf A}, {\bf B}, {\bf c})}$ breaks the steerability of $\rho_r$, then we must have ${\bf V}(r) + i({\bf 0}\oplus {\bf \Omega})\geq {\bf 0}$ for all $r$. In particular, if ${\bf z}\in{\rm Ker}\, {\bf C}$, then 
\begin{equation*}
\lim_{r\to\infty} \overline{{\bf z}}^T {\bf V}(r) {\bf z} = \overline{{\bf z}}^T  ({\bf \Omega B}{\bf \Omega}^T \oplus i {\bf \Omega} ){\bf z}\geq 0 \, .
\end{equation*}
By writing ${\bf z}= \left(\begin{array}{cc} {\bf z}_1 \\ {\bf z}_2  \end{array}\right)$ we find that ${\bf z}\in{\rm Ker}\, {\bf C}$ if and only if ${\bf z}_2 ={\bf \Omega ZA}{\bf \Omega}^T {\bf z}_1$. For such ${\bf z}$ we obtain
\begin{align*}
0& \leq \overline{{\bf z}}^T  ({\bf \Omega B}{\bf \Omega}^T \oplus i {\bf \Omega} ){\bf z} = \overline{\bf z}_1^T {\bf \Omega B}{\bf \Omega}^T{\bf z}_1 + \overline{{\bf z}}_2^T i{\bf \Omega} {\bf z}_2\\
&=  \overline{\bf z}_1^T{\bf \Omega} ({\bf B} + i {\bf A}^T {\bf Z}^T {\bf \Omega}^T {\bf \Omega}{\bf \Omega ZA}  ){\bf \Omega}^T {\bf z}_1\\
&=  \overline{\bf z}_1^T{\bf \Omega} ({\bf B} - i {\bf A}^T{\bf \Omega A}  ){\bf \Omega}^T {\bf z}_1.
\end{align*}
Since ${\bf z}_1$ is arbitrary and ${\bf \Omega}$ is invertible, we must have ${\bf B} - i {\bf A}^T{\bf \Omega A}\geq {\bf 0} $. Using Prop. \ref{prop:gaussian_ibc}, we conclude that $\Lambda_{({\bf A}, {\bf B}, {\bf c})}$ is Gaussian incompatibility breaking.
\end{proof}

\section{Conclusions}
We have characterised Gaussian channels which map the total set of all Gaussian observables into a set which is jointly measurable with a Gaussian joint observable. We call such channels Gaussian incompatibility breaking. We have shown that each entanglement breaking Gaussian channel is also Gaussian incompatibility breaking, but the converse does not hold. Finally, we have proven a connection to Gaussian EPR-steering by showing that Gaussian incompatibility breaking channels are exactly those channels which, when applied to one component of an arbitrary bipartite Gaussian state, make the state non-steerable with Gaussian measurements.

\section*{Acknowledgments}
JS acknowledges support from the Italian Ministry of Education, University and Research (FIRB project RBFR10COAQ). JK acknowledges support from the EPSRC project EP/J009776/1.

\end{document}